\documentclass[11pt]{article}

\usepackage{times}
\usepackage{amsmath,amssymb}
\usepackage{graphicx}

\newtheorem{theorem}{Theorem}[section]
\newtheorem{lemma}[theorem]{Lemma}

\newtheorem{satz}[theorem]{Proposition}
\newtheorem{defn}[theorem]{Definition}

\newcommand{\sq}{\hbox{\rlap{$\sqcap$}$\sqcup$}}
\newcommand{\qed}{\hspace*{\fill}\sq}
\newenvironment{proof}{\noindent {\bf Proof.}\ }{\qed\par\vskip 4mm\par}
\newenvironment{proofof}[1]{\bigskip \noindent {\bf Proof of #1:}\quad }
{\qed\par\vskip 4mm\par}

\newcommand{\bO}{\ensuremath{\mathcal{O}}}

\newcommand{\poly}{\ensuremath{\mathtt{poly}}}

\newcommand{\setN}{\ensuremath{\mathbb{N}}}

\newcommand{\setGW}{\ensuremath{\mathbb{GW}}}
\newcommand{\setG}{\ensuremath{\mathbb{G}}}
\newcommand{\setT}{\ensuremath{\mathbb{T}}}

\newcommand{\NH}{\ensuremath{\mathtt{N}}}
\newcommand{\Id}{\ensuremath{\mathtt{Id}}}
\newcommand{\Inp}{\ensuremath{\mathtt{Inp}}}

\newcommand{\LD}{\ensuremath{\mathcal{LD}}}
\newcommand{\NLD}{\ensuremath{\mathcal{NLD}}}
\newcommand{\NLDf}[1]{\NLD(\bO(#1))}
\newcommand{\NLDn}{\ensuremath{\mathcal{NLD}^{\#n}}}
\newcommand{\LOCAL}{\ensuremath{\mathcal{LOCAL}}}

\newcommand{\dist}{\ensuremath{\mathtt{dist}}}

\newcommand{\V}{\ensuremath{\mathcal{V}}}

\newcommand{\w}{\ensuremath{\mathtt{w}}}

\newcommand{\langSet}[2]{ \ensuremath{ \left\{ #1 \left| \,\begin{matrix} #2\end{matrix}\right.\right\}  }}

\newcommand{\algorithm}[1]{#1}

\newcommand{\lang}[1]{\ensuremath{\mathtt{#1}}}

\newcommand{\coMaximumMatching}{\overline{\lang{MaximumMatching} }}
\newcommand{\HasPerfectMatching}{\lang{HasPerfectMatching}}

\newcommand{\kDominatingSetNotFixed}{\lang{k-DominatingSet}^{NotFixed}}

\newcommand{\DomaticNumber}{\lang{DomaticNumber}}

\newcommand{\kColorable}{\lang{k-Colorable}}
\newcommand{\kEdgeColorable}{\lang{k-EdgeColorable}}

\newcommand{\Max}{\lang{Max}}
\newcommand{\Min}{\lang{Min}}

\newcommand{\Avg}{\lang{Avg}}
\newcommand{\LogicalOr}{\lang{LogicalOr}}
\newcommand{\LogicalAnd}{\lang{LogicalAnd}}
\newcommand{\Mode}{\lang{Mode}}

\newcommand{\Tree}{\lang{Tree}}

\newcommand{\AvgDeg}{\lang{AvgDeg}}
\newcommand{\MostDeg}{\lang{MostDeg}}
\newcommand{\Clique}{\lang{Clique}}
\newcommand{\Bipartite}{\lang{Bipartite}}

\newcommand{\kConsensus}{\lang{k-Set Consensus}}
\newcommand{\kSmallest}{\lang{k-SmallestElements}}

\newcommand{\FixedPointFreeSymmetryOnTrees}{\lang{FPFSymmetryOnTrees}}
\newcommand{\EqualSizePartition}{\lang{EqSizePartition}}

\newcommand{\NonEmptySetIntersection}{\lang{NonEmptySetIntersection}}

\begin{document}

\begin{titlepage}

\title{Hierarchies in Local Distributed Decision\footnote{This work was partially supported by the German Research Foundation (DFG) within the Collaborative Research Centre ``On-The-Fly Computing'' (SFB 901), by the EU within FET project MULTIPLEX under contract no.\ 317532.} \\
}

\author{Friedhelm Meyer auf der Heide\\ 
   fmadh@uni-paderborn.de \\
   \and
   Kamil Swierkot \\
   kamil@mail.uni-paderborn.de \\
   \and
      Heinz Nixdorf Institute \& Department of Computer Science \\
       University of Paderborn\\
  Paderborn, Germany\\
   }

\date{}

\maketitle \thispagestyle{empty}


\begin{abstract}
We study the complexity theory for the local distributed setting introduced by Korman, Peleg and Fraigniaud in  their seminal paper \cite{peleg2010ld}.  They have defined three complexity classes $\LD$ (Local Decision), $\NLD$ (Nondeterministic Local Decision) and $\NLDn$. The class $\LD$ consists of all languages which can be decided with a constant number of communication rounds. The class $\NLD$ consists of all languages which can be verified by a nondeterministic algorithm with a constant number of communication rounds. 
In order to define the nondeterministic classes, they have transferred the notation of nondeterminism into the distributed setting by the use of certificates and verifiers. The class $\NLDn$ consists of all languages which can be verified by a nondeterministic algorithm where each node has access to an oracle for the number of nodes.  They have shown the hierarchy $\LD\subsetneq\NLD\subsetneq\NLDn$. 

Our main contributions are strict hierarchies within the classes defined by Korman, Peleg and Fraigniaud. We define additional complexity classes: the class $\LD(t)$ consists of all languages which can be decided with at most $t$ communication rounds.
The class $\NLDf{f}$ consists of all languages which can be verified by a local verifier such that the size of the certificates that are needed to verify the language are bounded by a function from $\bO(f)$.  Our main result is the following hierarchy within the nondeterministic classes:
\begin{equation*}
\begin{matrix}
\LD&\subsetneq& \NLD(\bO(1)) & \subsetneq&\NLD(\bO(\log n) )&\subsetneq&\NLD(\bO(n)) \\ & \subsetneq&\NLD(\bO( n^{2}))&\subseteq&\NLD(\bO(n^2 + |w|)) &=& \NLD
\end{matrix}.
\end{equation*}
In order to prove this hierarchy, we give several lower bounds on the sizes of certificates that are needed to verify some languages from $\NLD$. For the deterministic classes we prove the following hierarchy:
\begin{equation*}
\LD(1)\subsetneq \LD(2) \subsetneq\LD(3) \subsetneq \ldots \subsetneq\LD.
\end{equation*} 
\end{abstract}

\bigskip

\centerline{{\bf Keywords}: Distributed Algorithms, Distributed Complexity Theory, Nondeterministic Local Decision}

\end{titlepage}


\section{Introduction}


   Over the last decade a lot of research has been done in order to find efficient algorithms for several problems in the distributed framework.  Thereby, researchers have achieved impressive  positive and impossibility results. An excellent overview is given by Suomela in his  paper \cite{suomela2009survey}.      
 But just recently, researchers have started to evolve  a complexity theory for the distributed framework: Korman, Peleg and Fraigniaud \cite{peleg2010ld} have  introduced a complexity theory in the $\LOCAL$ model. Like in traditional complexity theory, the focus is on decision problems.  
They have defined complexity classes which  help us to classify languages according to the hardness of solving them locally. The most restrictive complexity  class $\LD$  consists of all problems which can be decided by a local distributed algorithm, i.e. a distributed algorithm that only needs a constant number of rounds.
 
Korman, Peleg and Fraigniaud have introduced a notion of nondeterminism on the basis of certificates and local verifiers. The certificates are not allowed to depend on the nodes' identifiers. A local verifier is a local distributed algorithm which is able to use a certificate to decide whether an input belongs to the considered language, i.e., there has to exist a certificate which leads the algorithm to an accepting computation if and only if the input belongs to the language. The class $\NLD$  consists of all problems which can be decided by a local verifier.  Even though they have  already proven several structural properties of their local complexity theory in their seminal paper \cite{peleg2010ld}, there are still several open problems. 
 
Our contribution consists of two hierarchies for this complexity theory.  In the deterministic case, we define the class $\LD(t)$ which consists of all languages that can deterministically be decided in at most $t$ rounds. We  prove the hierarchy $\LD(t) \subsetneq \LD(t+1)$ for all $t\in\setN$.
In the nondeterministic case, we  refine the class $\NLD$ by defining classes $\NLDf{f}$ which consist of all languages which can be verified by a verifier with certificates of size at most $\bO(f(n,|w|))$. Here $n$ denotes the number of nodes of the input graph, $w = (w_1, \cdots, w_n$ the input assigned to the nodes of the graph, $|w|$ the maximum length of the $w_i$'s, and $f$ some complexity function in $n$ and $w$. We will show:
\begin{equation*}
\begin{matrix}
 \NLD(\bO(1))  \subsetneq \NLD(\bO(\log n) )\subsetneq \NLD(\bO(n))  \subsetneq \NLD(\bO( n^{2})) \subseteq \NLD(\bO(n^2 + |w|))=\NLD
\end{matrix}.
\end{equation*}

G\"o\"os and Suomela have already given a hierarchy for a different distributed complexity theory in \cite{locallyCheckableProof}. In particular, they studied locally checkable proofs. Basically, nodes get a proof of a given size assigned as additional inputs which can be used to verify a property of the network. In contrast to Korman, Peleg and Fraigniaud's model, these additional inputs are allowed to depend on the given identifiers. This allows us to verify each pure graph property  with proofs of size $\bO(n^2)$ \cite{locallyCheckableProof} which does not hold in our model \cite{peleg2010ld}. Nevertheless, we will use a technique introduced by G\"o\"os and Suomela to prove some of our lower bounds. The concept of locally checkable proofs is closely related to proof labeling schemes which have been studied in several papers for different problems, e.g. \cite{korman2010proof,fraigniaud2011locality,korman2006distributed,Korman:2005:PLS:1073814.1073817}.

Our paper is structured as follows: in Section \ref{sec:Pre}, we  describe the complexity theory in the distributed framework introduced  by Korman, Peleg and Fraignaud. In Section \ref{sec:ldhierarchy}, we will present the hierarchy for the deterministic case. In Section \ref{NLDLifts}, we provide an alternative  characterization of $\NLD$ which allows to design nondeterministic local algorithms with small certificates. This technique is then used to design algorithms needed for our separations of complexity classes.  In Section  \ref{sec:nldhierarchy}, we prove the hierarchy for the nondeterministic case by complementing the just mentioned upper bounds  on the certificate size with asymptotically tight lower bounds.

%

\section{Preliminaries}\label{sec:Pre}

In this section we will briefly introduce the necessary notations and  introduce the complexity theory for the distributed setting introduced by Korman, Peleg and Fraigniaud. 

\subsection{The Model}

We will focus on  the well known $\LOCAL$ model \cite{peleg:local,suomela2009survey}. The model is based on message passing with  synchronous fault-free point-to-point communications  described by a connected, undirected graph $G=(V,E)$. We allow messages  of unbounded size. The nodes of the graphs are considered as  processors and the  edges between them are bidirectional communication links.  We often use the term \textit{network} instead of \textit{graph}. 
We assume unique identifiers for each node.  But we do not want identifiers to encode any information. Therefore, we use the concept of identifier assignments \cite{peleg2010ld}: an \textbf{identifier assignment} is a map $\Id: \V \rightarrow \{1,\ldots, \poly(|\V|)\}$ such that for all vertices $u,v \in \V$, $u\neq v$,  $\Id(u) \neq \Id(v)$ holds. 
 Initially, a processor knows its own identifier and the communication links it can use. The identifiers or  any  information  stored at  the neighbors are initially unknown.

As our main concern is to study the complexity of locality,  our complexity measure  is the number of communication rounds that is needed by the algorithm. A processor can communicate with all its neighbors within one round. We allow arbitrarily complex computations in the local computation step.  The running time of an algorithm that  is the maximum number of communication rounds that are needed by the algorithm for all possible identifier assignments and all possible inputs. We call an algorithm \textit{local} if the number of rounds for each network is bounded by a constant. 

This model is strongly connected to the fundamental notion of $t$-neighborhoods. For  $t \in \setN$,  let 
 $\NH_{t}(v)$ denote the  subgraph induced by  all nodes with distance at most $t$ to $v$, 
 except for the edges between nodes in distance exactly $t$ to $v$. The $t$-neighborhoods include the nodes' identifiers.
We call two neighborhoods $\NH^{G}(v)$ and $\NH^{G'}(v')$ isomorphic, denoted by $\NH^{G}(v) \cong \NH^{G'}(v')$, if there exists a graph automorphism $\varphi: \NH^{G}(v) \rightarrow \NH^{G'}(v')$ such that $\varphi(v)=v'$ holds. 

\subsection{Distributed Complexity Theory}

The  complexity theory by Korman, Peleg and Fraignaud for the distributed setting focuses on decision problems.  
 A \textbf{configuration} is a pair $(G,w)$ where $G$ is a connected network, and every node $v \in \V$ is assigned as its local input a binary string $w(v)\in\{0,1\}^{*}$.
In order to transfer the concept of languages to our model, we define the set
$\setGW = \langSet{(G,w)}{ G\text{ is a connected network },w\in(\{0,1\}^{*})^{|\V(G)|} }$
as all possible configurations.  Then 
a \textbf{distributed language} is a decidable subset of $\setGW$.
For example, the following two languages which involve trees will be of interest
 $$
 \Tree_{t} = \langSet{(G,\epsilon)}{ G \text{ is a tree that contains a node }  v\in \V\\ \text{ such that,} \text{ for all } v' \in \V, \quad\dist_{G}(v,v') \leq t \text{ holds}}.
 $$
We say that a distributed algorithm \textbf{decides} a distributed language $L$ if, for each identifier assignment and every configuration $(G,w)$, all nodes accept $(G,w)$ if and only if $(G,w) \in L$ holds.
Although  nodes can use identifiers during their computations, nodes cannot assume  any encoded information in the identifiers. 

It is worth to note that every language can be decided by an algorithm that is allowed to use as many communication rounds as the diameter of the network. Therefore, we need to bound the number of communication rounds to define any meaningful complexity class. Since we are interested in the difficulties that arise from locality, the running time of an algorithm must be upper bounded by a constant. 
\begin{defn}
$\LD(t)$ is the class of all  languages that can be decided by a local algorithm with at most $t$ rounds.
\end{defn}
\begin{defn}
$\LD:=\bigcup_{t\in\setN}\LD(t)$ . 
\end{defn}

 Korman, Peleg and Fraigniaud  introduced nondeterminism in the distributed setting by giving each node its own certificate.  
Their definition of a (local) verifier is as follows: 
a \textbf{verifier} is a distributed algorithm that gets, additionally  to the configuration $(G,w)$, a certificate vector $c=(c(v), v\in V)$, where $c(v) \in \{0,1\}^{*}$ is known to node $v$.  We say a verifier \textbf{verifies} a language $L$ if the following holds:
\begin{itemize}
  \item If $(G,w) \in L$ holds, then there exists a certificate such that for each identifier assignment all nodes accept the configuration.
  \item If $(G,w) \notin L$ holds, then for all  certificates and  for each identifier assignment there exists a node which rejects the configuration.
\end{itemize}
The definition tries to reduce the impact of identifiers. Thus, certificates \textbf{may not depend on the identifier assignment} as a certificate must work on any assignment. The nondeterministic class is defined as follows:
\begin{defn}
 $\NLD$ is the class of all distributed languages that can be verified by a local verifier.
\end{defn}
We define the size $|c|$ of a certificate $c$ as $\max_{v\in\V} |c(v)|.$
Moreover, for a verifier $A$ and a configuration $(G,w) \in L$, let $c_{(G,w)}^{A}$ denote the size of a certificate of minimum size which leads to an accepting computation. 
We measure the certificate size as a function $f$ in the number $n$ of nodes of $G$ and the length $|w|$ of the input.

 As usual we require them to be monotone increasing in both arguments. We call a  local verifier \textit{$f$-bounded}, if, for each configuration $(G,w)$ for an $n$-node graph $G$, $c_{(G,w)}^{A} \leq f(G,w)$ holds.
With this notation, we can introduce subclasses of $\NLD$ with bounded certificate sizes.
\begin{defn}
Let $F$ be a set of functions $f:\setN^2 \rightarrow \setN$. Then $\NLD(F)$ is the class of all distributed languages $L$ that can be verified by a local verifier $A$ that is $f$-bounded for some $f \in F$. 
\end{defn}
Clearly,  $F \subseteq F'$ implies $\NLD(F) \subseteq \NLD(F')$.
As noted above, we will use the following sets of functions: $F=\bO(\log n)$, $F=\bO(n)$, $F=\bO(n^2)$ and $F=\bO(n^2 + |w|)$.


\section{The Deterministic Hierarchy}\label{sec:ldhierarchy}

The following theorem presents a strict hierarchy within $\LD$. 
\begin{theorem}\label{thm:LD:Hierarchy}
 Let $t$ be a positive integer. Then we have 
 $ \LD(t) \subsetneq \LD(t+1).$
\end{theorem}
In order to prove this result, we will consider the family of languages $\Tree_{t}$ introduced in Section 2.2.
Firstly, we give an algorithm for the language.
\begin{satz}\label{satz:insideLD:Tree}
For every positive integer $t$ the language $\Tree_{t}$ belongs to $\LD(t+1)$. 
\end{satz}
\begin{proof}
Let $t$ be a positive integer and let $(G,\epsilon)$ be a configuration. We consider the following local algorithm with running time $t+1$. 
\algorithm { 
Each node $v$ performs the following steps: 
\begin{enumerate}
 \item Calculate $\NH_{t+1}(v)$.
 \item Accept if $\NH_{t+1}(v)$ is a tree and 
 $ \max_{u,u' \in \NH_{t+1}(v)} \dist_{G}(u,u') \leq 2t $
 holds. Otherwise reject the input.
\end{enumerate}
}
We first show that, if the algorithm accepts $(G,\epsilon)$, the diameter of the graph is at most $2t$. For the sake of  contradiction, we assume that this is false. Let $u,u'$ be two nodes with $\dist_{G}(u,u') = 2t+1$. Then let $u''$ be a node on the shortest path from $u$ to $u'$ with $\dist_{G}(u,u'') = t+1$ and $\dist_{G}(u'',u')= t$. Thus, we have $u',u \in \NH_{t+1}(u'')$. According to the second step of the algorithm, this node will reject the input. Therefore, the diameter of the graph is at most $2t$.

This implies that there is a node $v^*$  that has distance at most $t$ to all other nodes. Then step 2 makes sure that the graph is a tree  because otherwise $v^*$ would have rejected the input. Thus, the algorithm accepts the configuration $(G,\epsilon)$ if and only if $(G,\epsilon) \in \Tree_{t}$ holds.
\end{proof}
The following proposition implies that the algorithm is optimal.
\begin{satz}
 Every local algorithm for the language $\Tree_{t}$ needs at least $t+1 \in\setN$ rounds.
\end{satz}
\begin{proof}
 For the sake of contradiction, we assume that there is a local algorithm which decides $\Tree_t$ in $t$ rounds. Consider the path $P=(v_{1},\ldots, v_{2t+1})$ of length $2t$. Then we have $(P,\epsilon) \in \Tree_{t}$ because every node in the path has distance at most $t$ to the node $v_{t+1}$. The algorithm must accept the configuration $(P,\epsilon)$. In particular, the node $v_{t+1}$ accepts the input. 

Now consider a cycle $C$ of size $2t+2$.  Then the configuration $(C,\epsilon)$ does not belong to the language $\Tree_{t}$. But the $t$-neighborhood of every node of the cycle is isomorphic  to the $t$-neighborhood of $v_{t+1}$ in $(P,\epsilon)$. Therefore, the algorithm must accept the configuration $(C,\epsilon)$. This is a contradiction. Thus, a local algorithm with running time $t$ cannot decide the language $\Tree_{t}$. 
\end{proof}
These results imply Theorem \ref{thm:LD:Hierarchy} as follows:
for every   $t \in \setN$ we have $\Tree_{t} \in \LD(t+1)$ and  $\Tree_{t} \notin \LD(t)$. Thus, we have $\LD(t) \neq LD(t+1)$. Since $\LD(t) \subseteq \LD(t+1)$ holds, we have $\LD(t)\subsetneq \LD(t+1)$.


\section{The Nondeterministic Hierarchy}

In the remainder of this paper will show our main result; a hierarchy within NLD. 
\begin{theorem}\label{thm:separation:NLD}
The following hierarchy  holds in $\NLD$:
\begin{equation*}
 \begin{matrix}
 \LD & \stackrel{(i)} {\subsetneq}& \NLDf{1} &\stackrel{(ii)} {\subsetneq} & \NLDf{\log n} & \stackrel{(iii)} {\subsetneq}& \NLDf{n} \\
        &  \stackrel{(iv)} {\subsetneq}&  \NLDf{n^{2}} &\stackrel{(v)} {\subseteq}& \NLDf{n^{2} + |w|} &\stackrel{(vi)} {=} &\NLD
  \end{matrix}.
  \end{equation*}
\end{theorem}
\noindent\textbf{Survey of the proof:}   All inclusion in $(i)$ to $(v)$ are clear. We have to show the corresponding inequalities. The one in $(i)$ is  shown in \cite{peleg2010ld}. 
For proving the inequalities in $(ii)$, $(iii)$ and $(iv)$, we define separating languages. For (ii), the language 
 $$\Tree = \bigcup_{t\in\setN} \Tree_t$$is used. The upper bound- $\Tree \in \NLD(\bO(\log(n)))$ - is shown in \cite{peleg2010ld}. We will present the following lower bound. 
 \begin{lemma}\label{thm:tree:logn:bound}
 Every local verifier for $\Tree$  needs certificates of size $\Omega(\log n)$.
\end{lemma}
 For proving $(iii)$ we use the language$$
 \FixedPointFreeSymmetryOnTrees = \langSet{(G,\epsilon)}{ G \text{ is a tree and there is   a graph automorphism }\\ \lambda: \V\rightarrow\V   \text{ such }\lambda(v) \neq v \text{ for all }v\in V} .
 $$This language was already used in \cite{locallyCheckableProof} in order to prove a separation result in the setting of locally checkable proofs.  We will show upper and lower bounds for the certificate  size needed, yielding the separation result $(iii)$. 
 
\begin{lemma}\label{satz:insideNLD:FixedPointFreeSymmetryOnTrees}
Every decidable language $L \subseteq Trees$ (thus in particular $\FixedPointFreeSymmetryOnTrees$) can be decided by a local verifier with certificates of size $\bO(n)$. 
\end{lemma}
\begin{lemma}\label{thm:lowerbound:fixedpointfreesymmetryontrees}
Every local verifier for $\FixedPointFreeSymmetryOnTrees$  needs certificates of size  $\Omega(n)$.
\end{lemma}The separation $(v)$ is demonstrated using the language 
$$
\EqualSizePartition = \langSet{(G,\w)}{w(v) \in \{0,1\},  |\{v \in \V |Êw(v) = 0\} | = | \{v \in \V |Êw(v) = 1\} |} .
$$To the best of our knowledge, this language was not considered before for separation results. We prove upper and lower bounds on the certificate size needed, yielding the separation result (v). \begin{lemma}\label{satz:insideNLD:EqualSizePartition}
  $\EqualSizePartition$ can be decided by a local verifier with certificates of size $\bO(n^{2})$.
\end{lemma}\begin{lemma}\label{thm:lowerbound:eqsizepartition}
Every local verifier for  $\EqualSizePartition$ needs certificates of size  $\Omega(n^{2})$.
\end{lemma}
The proof of $(vi)$ follows from \begin{lemma}\label{NLDn2w}
 Let $L \in \NLD$  be a distributed language from $\NLD$. Then $L$ is decidable by a local verifier with certificates of  size at most $\bO(n^{2} + |w|)$. Thus, $\NLD = \NLD(\bO(n^{2} + |w|))$ holds.
\end{lemma}
All new upper bounds, i.e. the Lemmata \ref{NLDn2w}, \ref{satz:insideNLD:FixedPointFreeSymmetryOnTrees}, and \ref{satz:insideNLD:FixedPointFreeSymmetryOnTrees} are presented in Section \ref{NLDLifts}. For this, we present extensions of the notion of lifted configurations introduced in \cite{fraigniaud2010local}. Insights into this concept allow a characterization of $\NLD$ yielding Lemma \ref{NLDn2w}, and the upper bounds in Lemmata \ref{satz:insideNLD:FixedPointFreeSymmetryOnTrees} and \ref{satz:insideNLD:EqualSizePartition}. The lower bounds, i.e.  the Lemmata \ref{thm:lowerbound:fixedpointfreesymmetryontrees} and \ref{thm:lowerbound:eqsizepartition}, are inspired by an approach use in \cite{locallyCheckableProof}. This approach is as follows:   One  construct a graph from two graphs  which are connected by a path. Some of these constructed graphs belong to the language and other do not belong to the language. Then one only considers the certificates of some  nodes on the path to construct a certificate for a negative graph leading to an accepting computation.


\section{The Upper Bounds}\label{NLDLifts}

In this section, we will describe and extend the notation of lifts introduced in \cite{fraigniaud2010local}. It will allow us 
characterize languages from $\NLD$ in a way  that yields the $\bO(n^{2} + |w|)$ upper bound for the certificate sizes of all languages in $\NLD$. This proves  Lemma \ref{NLDn2w}.

We call a configuration $(G,w)$  a \textbf{lifted configuration} of a configuration $(G',w')$, if there exists a surjective map $\lambda: \V(G) \rightarrow \V(G')$ such that, on input $(G,w)$, every node $v$ accepts  the certificate $c(v) = ((G',w'), \lambda(v))$ with the following verifier:
\algorithm {  
Each node $v$ performs the following:
\begin{enumerate}
 \item Check whether all neighbors have the same entry $(G',w')$ and reject otherwise.
 \item Check whether for all neighbors $u \in \NH^G(v)$, both    $\lambda(u) \in \NH^{G'}(\lambda(v))$ and $|\NH^G(v)| = |\NH^{G'}(\lambda(v))|$ hold.
 \item Check whether $w'(\lambda(v)) = w(v)$ holds.
\end{enumerate}
}
The configuration $(G',w')$ is called the lift of the configuration $(G,w)$. An example is shown in Figure \ref{fig:lifts:examples}. Intuitively,  a lift must be a graph which preserves most characteristics of the lifted configuration and which can be embedded into the lifted configuration such that all nodes of the lift occur equally often in the lifted configuration. Korman, Peleg and Fraignauid used this technique to show  for the language $\lang{containment}$ that it belongs to $\NLD$ in \cite{fraigniaud2010local}. Moreover, they claimed that $\NLD$ must be strongly related to the lifts. But, they did not give any proof for this.

\subsection{A Verification Scheme of All Languages from $\NLD$}
Firstly, we will show the importance of lifts for the class $\NLD$, namely that all languages from $\NLD$ can be verified with this technique. 
We say that a configuration $(G,w)$ is a \textbf{t-lifted configuration} of a configuration $(G',w')$ if the lift algorithm checks for each $v\in \V$ whether all nodes in its $t$-neighborhood are correctly labeled according to $(G',w')$. We show that languages from $\NLD$ "are closed under $k$-lifts".
\begin{lemma}\label{satz:lift:NLD}
 Let $ L \in \NLD$. Then there exists a $t \in \setN$ such that, if $(G,w)$ is a $t$-lifted configuration of a configuration $(G',w') \in L$, then $(G,w)$ belongs to $L$. 
\end{lemma}
\begin{proof}
If $L\in\NLD$ holds, there exists a verifier $A$ that verifies $L$ in some number $t$ of rounds. Let $(G,w)$ be a $t$-lifted configuration of $(G',w')$ with  map $\lambda:\V \rightarrow \V'$ and $(G',w') \in L$. Since $(G',w')\in L$ holds, there exists a certificate $c'$ such that  $A$ accepts $(G',w')$ with the certificate $c'$. We assume $(G,w)\notin L$. We define the certificate $c$ by $c(v) = c'(\lambda(v))$. Then for each $v \in \V$ we have  $\NH_{t}^{G,w,c}(v) \cong\NH_{t}^{G',w',c'}(\lambda(v))$. Thus, $A$ accepts  $(G,w) \notin L$ with certificate $c$,  contradicting the correctness of the algorithm. Thus, $(G,w) \in L$ holds.
\end{proof}
Computing the length of the above certificates yields our  $\bO(n^{2} + |w|)$ upper bound on the required certificate size for languages in $\NLD$, and proves Lemma \ref{NLDn2w}.

\subsection{Linear Size Certificates for Fixed Point Free Symmetries on Trees }

Now, we want to know whether a graph possesses a fixed point free symmetry. In order to prove the containment to $\NLD$, we restrict the language to trees. In order to find a verifier with smallest certificate size,  we will use the notation of lifts and use a compressed description of the lift. Since the lift needs to be a tree, it can be encoded  with $2n$ bits \cite{zaks1980lexicographic}.

Before we can prove the upper bound, we need to study the relationship of lifts and lifted configurations.
Therefore, we assume  $\V(G') = \{1,\ldots,|\V(G')|\} \subseteq \setN$ and  define the set
$ \V_\lambda(i) = \{v \in V | \lambda(v) = i \} $
of $i\in\{1,\ldots, |\V|\}$ labeled nodes.
We are ready to state some simple properties.
\begin{lemma}\label{satz:lift}
Let $(G,\w)$ be the lifted configuration of a configuration $(G',\w')$. Then we have:
\begin{enumerate}
  \item For every $i \in \lambda(\V)$ and every $u \in V_\lambda(i)$ we have  $\w(u) = \w'(i)$ and $\deg_{G'}(i) = \deg_{G}(u)$. 
  \item For every $v \in \V(G)$ the neighborhoods $\NH^{G,w}(v)$ and $\NH^{G',w'}(\lambda(v))$ are isomorphic. 
  \item There exists an $l \in \setN$ such that, for every $i \in \lambda(\V)$,  $|\V(i)| = l$ holds.
\end{enumerate}
\end{lemma}
\begin{proof}
  We omit the obvious proofs of the first claims. In order to conclude Claim 3, we  show the following by induction on the length of the  path.
 
 $\bullet$ If there is a path from $i$ to $j$ in $G'$, then  $|\V(i)| = |\V(j)|$ holds.
  
  \noindent\textbf{Induction Basis:} Let $i,j$ be neighbors in $G'$. Let $v\in \V(i)$ be a node of $G$. Then step 3 of the verifier implies that there is $u \in \NH(v)$ such that $\lambda(u) = j$ holds. Obviously, only one node fulfills this property. By the symmetry of this argument, the map $\varphi_{i,j}: \V(i)  \rightarrow  \V(j)$ defined by  $\varphi_{i,j}(v)= u \in N(v), \lambda(u) = j$
  is well defined and bijective. Thus, we have $|\V(j)| = |\V(i)|$.
     
  \noindent\textbf{Induction Step:}
   Let $i,j$ be connected by a shortest path of length $k+1$ in $G'$. Let $i,i_1,\ldots,i_k,j$ denote a shortest path between them. Then the induction hypothesis implies
   $ |\V(i)| = |\V(i_1)| = \ldots = |\V(i_k)|$
   and 
   $ |\V(i_1)| = \ldots = |\V(i_k)| = |\V(j)|.$ 
 Thus, we have $|\V(i)| = |\V(j)|$.

 Now the claim follows from the fact that $G'$ is a simple, connected graph, and, thus, there is a path between every pair of nodes of $G'$.  All sets $\V(i)$ have the same cardinality and we can define $l = |\V(i)|$.
\end{proof}
After we have seen the properties of the lifted configuration according to the $i$-labeled sets, we will focus on the relationship of the inputs $w$ and $w'$. This properties will be of interest when we will use the notation of lifted configurations to show the containment of some languages  to the class $\NLD$.
The following lemma will help us to prove claims on the  topological structure of $G$ and $G'$.
\begin{lemma}\label{lem:lift:map}
Let $G$ be a lifted configuration of $G'$. Then there exists a $k\in \setN$ such that for every $i \in \lambda(\V)$ there exists a map $\varphi_{i}:\V\rightarrow \{1,\ldots, k\}$ such that 
 for each $ j,l \in \{1,\ldots, k\}$  $|\varphi_{i}^{-1}(\{j\})| = |\varphi_{i}^{-1}(\{l\})|$ holds and  $G[\varphi_{i}^{-1}(\{j\})]$ is connected.
  Moreover, for each $j \in \{1,\ldots, k\} $ $\Id(\lambda(G[\varphi_{i}^{-1}(\{j\})])) = \lambda(\V)$ holds.
 \end{lemma}
 
 \begin{proof}
 We define $k = |\V(i)|$ which is,  by Proposition \ref{satz:lift}, indepent of $i$. We can assume that $k > 1$ holds.  Define weights on the edges of $G'$ such that all shortest paths from $i$ to every other node have unique lengths. Let $v_{1},\ldots,v_{k}$ be the nodes of $\V(i)$. Moreover, let $A^{m}_{j}$ denote the set of all nodes which have been assigned to partition $j$ in the $m$-th inductive step, where the sets are inductively defined as follows: we define $A^{1}_{j} = \{v_{j}\}$. Moreover, we define $A^{l+1}_{j} = A^{l}_{j} \cup \{v^{j}\}$, where $v^{j} \in \NH(A^{l}_{j})\setminus A^{l}_{j}$ is minimal with $\dist_{G'}(\lambda(v_{j}), \lambda(v^{j}))$. We show by induction that the defined sets fulfill the desired properties.
  
  \noindent\textbf{Induction Basis:}
  All sets $A^{1}_{j}$ have size one and consist of one node from the set $\V(i)$. Thus, all $k$ sets are disjoint and have the same cardinality. Since all $G[A^{1}_{j}]$ consist of only one node, the vertex induced graphs are connected. Furthermore, for all $ j \in \{1,\ldots,k\}$ it holds  $\lambda(\Id(G[A^{1}_{j}]) ) = \{i\}$. Hence, we have  $\lambda(\Id(G[A^{1}_{j}]) )= \lambda(\Id(G[A^{1}_{h}]) )$.
  
  \noindent\textbf{Induction Step:} Since every shortest path with the edge weights from $i$ to any other node is  unique  and $(G,w)$ is the lifted configuration of $(G',w')$, all nodes that have been chosen in step $l+1$, must have the same label. The induction hypothesis implies that  $\lambda(\Id(G[A^{l}_{j}]) )= \lambda(\Id(G[A^{l}_{h}]) )$ holds.  $\lambda(v^{j}) = \lambda(v^{h})$ implies that $\lambda(\Id(G[A^{l+1}_{j}]) )= \lambda(\Id(G[A^{l+1}_{h}]) )$ holds. Moreover, we have $|A^{l+1}_{j}| = |A^{l}_{j}| + 1$ as long as not all nodes are assigned. The induction hypothesis implies $A^{l}_{j} = A^l_i$  and, thus,  $A^{l+1}_{j}=A^{l+1}_{i}$holds.
  
  The induction hypothesis implies that all $G[A^{l+1}_{j}]$ are connected because all $G[A^{l}_{j}]$ are connected and the node $v^{j}$ is taken from the neighborhood of $A^{l}_{j}$ in $G$. Since  the cardinality of each set increases by one in an induction step, all nodes are assigned after step $m=\frac{|\V|}{k}$.
  Now we can define the map $\varphi_{i}: \V \rightarrow \{1,\ldots, k\}$ by
   $ \varphi_{i}(v) = j \text{ , if } v \in A^{m}_{j}.$
The map is well defined because all sets are disjoint. $\varphi_{i}^{-1}(\{j\}) = A^{m}_{j}$ implies that $\varphi_{i}$ obviously fulfills all properties.    
 \end{proof}
 
 Since the lift can be embedded in the lifted configuration such that each node of the lift occurs equally  often in the lifted configuration, there must be a cyclic arrangement. This is stated by the following lemma.
\begin{lemma}\label{satz:lift_cycle}
Let $(G',w')$ be a lift of $(G',w')$ such that $|\V(i)| >0$ holds. Then $G$ and $G'$ have  cycles. 
\end{lemma}
\begin{proof}
Fix an $i \in \lambda(\V)$. Let $\V(i) =\{v_{1},\ldots, v_{k}\}$ be the $i$-th node set.  Moreover, let $\varphi_{i}$ be a map like in Lemma \ref{lem:lift:map}. We define $P(v_{j}) = G[\lambda_{i}^{-1}(\{j\})]$ as the partition of $G$ with the nodes assigned to $v_{j}$. Since $G$ is connected, there is an edge between $P(v_{e})$ and $P(v_{j})$ for some $e,j \in \{1,\ldots,k\}$. Let this edge be denoted by $\{v'_{e},v'_{j}\}$ with $v'_{e} \in P(v_{e}), v'_{j} \in P(v_{j})$ and the labels $l = \lambda(v'_{j})$ and $m=\lambda(v'_{e})$.  Lemma \ref{lem:lift:map} implies $\lambda(P(v_{j})) = \lambda(V)$, and, thus, there is a node $v''_{j} \in P(v_{j})$ with label $m$. Since $G$ is a lift of $G'$, there is a neighbor $v''$ of $v''_{j}$ with label $l$. This node cannot belong to $P(v_{j})$. Thus, there are two cases: the node $v''$ belongs to $P(v_{e})$ or to some other $P(v_{j_{2}})$. If $v''$ belongs to $P(v_{e})$, we have a cycle in the graph $G$ because all partitions are connected. 
If $v''$ belongs to a $P(v_{j_{2}})$, we can apply the argument inductively and get a sequence $j,j_{2},\ldots, j_{s}$ of partitions such that $P(v_{j_{r}})$ is connected to $P(v_{j_{r+1}})$ through an edge with the label $l$ and $m$. 
 Since there are only finite partitions, there must be a back edge. Thus, we have found a cycle in $G$.
\end{proof}
Now, we are ready to turn to  the proof of Lemma \ref{satz:insideNLD:FixedPointFreeSymmetryOnTrees}.

\begin{proofof}{Lemma \ref{satz:insideNLD:FixedPointFreeSymmetryOnTrees}}
Consider an arbitrary decidable language $L \subseteq Trees$. For $(T,\epsilon)$ 
the certificate $c(v)=(T',\lambda(v))$ consists of an encoded tree $T'$ and the label $\lambda(v)$.  It has size $\bO(n)$ because the label has size $\bO(\log n)$ and the tree  is encoded with $2n$ bits. 
 We use the following verifier: Each node $v$ performs the following:
   \begin{enumerate}
       \item Reject if the certificate is not  the encoding of a tree $T'$  and  if $(T,\epsilon)$ is not a lift of $(T,\epsilon)$.
       \item Check  $(T',\epsilon) \in 
       L$ by brute force.
   \end{enumerate}
   
   The first step makes sure that only certificates are accepted which encode a tree $T'$ and that there are lifts of the input certificate. By  Lemma \ref{satz:lift_cycle}, the algorithm accepts if and only if  $T$ is a tree and $T = T'$ holds.
     Thus, in step 2, it is verified that $(T,\epsilon)$ belongs to $L$. Therefore,  the verifier is correct.
\end{proofof}

\subsection{Verifying Equally Sized Partitions with Quadratic Size Certificates}

 In order to prove the upper bound for the equally sized partitions, we need to further study the relationship between lifts and lifted configurations
\begin{lemma}\label{folgerung:lift}
 Let $(G,\w)$ be a lifted configuration of $(G',\w')$. Then there exists $l \in \setN$ such that we have:
 \begin{enumerate}
   \item $|\V| = l \cdot |\V'|$.
   \item For each input $\tilde{\w} \in \Inp( (G,w) )$ we have $|\{ v \in \V| w(v) = \tilde{\w}\}| = l\cdot |\{ v' \in \V'| w'(v') = \tilde{\w}\}|. $
 \end{enumerate}
\end{lemma}
\begin{proof}
 By Proposition \ref{satz:lift}, all $\V(i)$ have the same cardinality $l$. We denote the set of labels of $G'$ with $I$. Obviously, we have  $\V(i)\cap\V(j)=\emptyset$ for $i\neq j$. Since every $v\in \V$ is an element of a $\V(j)$, we have 
 \begin{equation*}
   |\V|  = |\bigcup_{i \in I} \V(i)| = \sum_{i\in I} |\V(i)| = \sum_{i\in I} l = l \cdot \sum_{v' \in V'} 1 = l \cdot |\V'| .
 \end{equation*} 
For the second claim, let $\tilde{w} \in \Inp((G,w))$ be an input of $(G,w)$. Then we have
\begin{equation*}
|\{ v \in \V| w(v) = \tilde{\w}\}| = |\bigcup_{i \in I} (\{ v \in \V| w(v) = \tilde{\w}\} \cap \V(i))| = \sum_{i\in I} |\{ v \in \V| w(v) = \tilde{\w}\} \cap \V(i))| .\end{equation*}
By Proposition \ref{satz:lift}, this value is  equal to
$
\sum_{i\in I \atop w'(i)= \tilde{w}} l = l \cdot |\{ v' \in \V'| w'(v') = \tilde{\w}\}|.
$
This proves the claim.
\end{proof}

\begin{proofof}{Lemma \ref{satz:insideNLD:EqualSizePartition}}
Let $(G,w)$ be the input. We use the lift certificate $((G',w'),\lambda(v))$. The certificate size is $\bO(n^{2})$.  The verifier works analogously to the one in the last proposition.
Again, the first step  ensures that  $(G',w')$ is  the lift of $(G,w)$.
If $(G',w')\in\EqualSizePartition$ holds, $w'$ defines the equal sized partitions $P'_{i} = \{v' | \w(v') = i\} \subseteq \V'$. By Corollary \ref{folgerung:lift}, all nodes of $\V(i)$ belong to the same partition. We define
$P_{0} = \bigcup_{i\in P'_{0}} \V(i)$ and $P_{1} = \bigcup_{j\in P'_{1}} \V(j).$
 $|\V(i)| = |\V(j)|=k$ implies the equal size of the partitions:
\begin{equation*}
 \begin{matrix}
 |P_{0}|  =  |\bigcup_{i\in P'_{0}} \V(i) | 
  = \sum_{i \in P'_{0}} |\V(i)| &
  =  k\cdot |P'_{0}|
 =  k\cdot |P'_{1}| 
 = \sum_{j \in P'_{1}} |\V(j)| &
 = |\bigcup_{j\in P'_{1}} \V(j) |
 = |P_{1}|\,.
 \end{matrix}
\end{equation*}
Since $\V(j) \cap \V(i) = \emptyset$ holds for $i\neq j$, the partitions are disjoint.
Thus,  the verifier works correctly.
\end{proofof}


\section{The Lower Bounds}\label{sec:nldhierarchy}

\subsection{A Logarithmic Lower Bound for the Certificate Size of Trees}\label{sec:log}
In this section, we prove a $\Omega(\log n)$ lower bound for verifying that a graph is a tree. The same technique can be used to prove logarithmic lower bounds for several other languages, c.f. the appendix. 

\begin{proofof}{Lemma \ref{thm:lowerbound:fixedpointfreesymmetryontrees}}
We count the possible labelings of $t$-neighborhoods with certificates, where the  $t$-neighborhoods  look like a path. In this case,  we denote the $t$-neighborhood of a node $v$ as $(a,b)_{v}$, where $a = (a_{1},\ldots, a_{t})$ and $b = (b_{1},\ldots, b_{t})$  are a path $a_{1},\ldots,a_{t},v,b_{1},\ldots,b_{t}$. Whenever we have a certificate $c$ for the nodes, we let $c(a,b)_{v} = (c(a_{1}), \ldots,c(a_{t}),c(b_{1}),\ldots, c(b_{t}))$ denote the vector of the certificates assigned to the nodes.   
We will show that we can find two nodes $u\neq v$ with $c(a,b)_{v}=c(a',b')_{u}$. Then we will create a new configuration which does not belong to the language and with  isomorphic $t$-neighborhoods.
 
Let $V$ be a local verifier that verifies  $\Tree$ in $t$ rounds. For the sake of contradiction, we assume  that $V$ needs certificates of size at most $k(n) = \frac{\frac{1}{2} \log (n) -(4t+4 +1)}{2t} - 1$. Then there are at most $2^{k(n)+1}$ different certificates which can be assigned to a node. We consider a path $P$ of size $ s(n) = (4t+4) \left((2^{k(n)+1})^{2t} + 1\right)2^{k(n) + 1}$ and allow certificates of size $k(n)$. A simple calculation shows that $s(n) < n$ holds.
Therefore, we allow the verifier to use larger certificates  than assumed. Thus, it is sufficient to come up with a  contradiction against the correctness of the verifier that uses these certificate sizes. 
We have $(P,\epsilon) \in \Tree$. Let $c$ be a certificate which leads to an accepting computation. Then let $\tilde{c}$ denote the most assigned certificate to the nodes. We define the set $ V_{\tilde{c}}^{c} = \{v \in \V\,|\, c(v) = \tilde{c}\}.$
By the definition of $s(n)$, we have $|\V_{\tilde{c}}^{c}| \geq (4t+4)\left({2^{k(n)+1}}^{2t} + 1\right)$ because there are at most $2^{k(n) + 1}$ different certificates. Thus, there must be at least $4t+4$ nodes $u,v\in  V_{\tilde{c}}^{c}$ with $c(a,b)_{u} = c(a',b')_{v}$ because there are at most ${2^{k(n)+1}}^{2t}$ different possibilities to label a $t$-neighborhood $(\tilde{a},\tilde{b})_{u'}$ of a node $u$. The factor $4t+4$ ensures that we can choose $u,v$ such that $(a,b)_{u}$ and $(a',b')_{v}$ do not share a node and that none of the nodes is an end node.

To come up with a contradiction, we can construct a new graph which does not belong to the language. Therefore, let $a,v,b,d,a',u,b'$ denote the subgraph of $P$ that starts at  the $t$-neighborhood of $u$ over the nodes $d=(v_{j_{1}},\ldots,v_{j_{l}})$ to the neighborhood of $v$, where we have $|a| = |b|=|a'|=|b'| = t$ and $l = |d|\geq0$, c.f. Figure \ref{fig:lowerbound:tree:cycleconstruction}. Then we can define a new graph $G$ which contains a cycle as follows:
 let $(a,b)_{v} = (a_{1},\ldots,a_{t},b_{1},\ldots,b_{t})$ and $(a,b)_{u} = (a'_{1},\ldots,a'_{t}, b'_{1}\ldots,b'_{t})$ denote the nodes in the $t$-neighborhood. Then let $d'$ be a copy of $d$. The graph consists of the paths $(a,b)_{v}$, $(a',b')_{u}$, $d$ and $d'$. Moreover, we add the edges $(b_{t},v_{j_{1}})$, $(v_{j_{l}},a'_{1})$, $(b'_{t},v'_{j_{1}})$ and $(v'_{j_{l}}, a_{1})$. The construction is visualized in Figure \ref{fig:lowerbound:tree:cycleconstruction}. Every node inherits the certificate from the corresponding node of $P$ with certificate $c$. Let $c'$ denote the new certificate for $G$. Then we have the following isomorphic $t$-neighborhoods for $ i\in\{1,\ldots,l\}$:
 $
  \NH^{(G,\epsilon),c'}(u)\cong \NH^{(P,\epsilon),c}(u), \NH^{(G,\epsilon),c'}(u)\cong \NH^{(P,\epsilon),c}(u), 
  \NH^{(G,\epsilon),c'}(a_{i}) \cong \NH^{(G,\epsilon),c'}(a'_{i}) \cong \NH^{(P,\epsilon),c}(a_{i}) 
  \NH^{(G,\epsilon),c'}(b_{i}) \cong \NH^{(G,\epsilon),c'}(b'_{i}) \cong \NH^{(P,\epsilon),c}(b_{i}) ,
  \NH^{(G,\epsilon),c'}(d_{i}) \cong \NH^{(G,\epsilon),c'}(d'_{i}) \cong \NH^{(P,\epsilon),c}(d_{i})  .
 $ 
 Thus, all nodes must  accept the input $(G,\epsilon)\notin \Tree$. This is a contradiction and, thus, certificates of size $k(n)$ are not sufficient. 
\end{proofof}

\subsection{A Quadratic Lower Bound for the Certificate Size of Equal Sized Partitions}\label{sec:super}

 We reconsider  the language of all graphs with an equal sized partition and study the family of connected unlabelled  graphs $\setG_{n}$ of size $n$. Since almost all graphs are connected \cite{west2001introduction}, for sufficient large $n$ we have
$ |\setG_{n}| \geq \frac{1}{2} 2^{{n \choose 2} }/n! \text{ and } \log( 2^{{n \choose 2} }/n!) = \Theta(n^{2}),$
where  $\frac{1}{n!}$ accounts for the removal of the isomorphic copies.  We have to take into account the isomorphic copies because these are only graphs with different identifier assignments. 

Now we will introduce the construction which will be the key for the proof.  We build  a graph from two graphs of  $\setG_{n}$ in the following way: let $G',G'' \in \setG_{n}$ be two graphs with $n$ nodes. Moreover, we fix one node $v' \in\V(G')$ and one node $v''\in \V(G'')$. Then the graph $G_{t}(G',v',i,G'',v'',j)$, $i,j\in\{0,1\}$, consists of the graphs $G'$, $G''$ and $P(t)$, where $P(t)$ denotes the path $v',v_{1},\ldots,v_{4t+4},v''$. The inputs are defined by: 
$
  \forall v\in\V(G')\, w(v) = i , 
  \forall v\in\V(G'') \,w(v) = j, 
  \forall i =1,\ldots, 4t+4\,w(v_{i}) = i \mod 2.
  $
  
  An example can be found in Figure \ref{fig:lowerbound:PartionGraph}. An important and obvious observation is that $G_{t}(G',v',i,G'',v'',j)$ is an element of $\EqualSizePartition$ if and only if $|\V(G')| = |\V(G'')|$ and $i\neq j$ hold. 

\begin{proofof}{Lemma \ref{thm:lowerbound:eqsizepartition}}
 For the sake of contradiction, we assume that there is a local verifier  that verifies the language in $t\in\setN$ rounds. Let $n\geq 4t+4$ be a positive integer  and define $g(n) = \log(2^{{n \choose 2} }/n!)$. We will show that the algorithm needs certificates of size at least $ \frac{g(n)}{2\cdot 10 \cdot(2t+2)}$. Therefore, we assume that the algorithm needs  certificates of size less than $ \frac{g(n)}{2\cdot 10 \cdot(2t+2)}$  per node. For all $G',G''Ê\in \setG_{n}$, we have $G_{t}(G',v',0,G'',v'',1)$, $G_{t}(G',v',1,G'',v'',0)  \in \EqualSizePartition$ and $G_{t}(G',v',0,G'',v'',0) \notin \EqualSizePartition$. The graphs have size $2n + 4t+4$. For sufficient large $n$ we have
$\frac{g(2n + 4t + 4)}{2\cdot 10 \cdot(2t+2)}\leq \frac{g(n)}{2  \cdot(2t+2)}$.
Thus, the nodes $v_{t+2},\ldots, v_{3t+3}$ of $P(t)$ have at most $\frac{g(n)}{2}$ bits and the nodes in the $t$-neighborhoods of these vertices have the same inputs in all instances of the graphs. Since this is  smaller than $\log \frac{|\setG_{n}|}{n!}$,  there are at most $|\setG_{n}|$ possible different certificates  that are used for the $2t+2$ nodes. Therefore, we have graphs $G',G'',G''',G'''' \in \setG_{n}$ and certificates $c_{1}$, $c_{2}$ such that the verifier accepts  $G_{t}(G',v',0,G'',v'',1)$ with certificate $c_{1}$ and that  the verifier accepts  $G_{t}(G''',v''',1,G'''',v'''',0)$ with certificate $c_{2}$. The certificates $c_{1},c_{2}$ assign the same certificates to the nodes of $P(t)$. 

In contrast, we can look at the graph $G_{c} = G_{t}(G',v',0,G'''',v'''',0)$ and define the certificate $c$ by: $
 \forall v \in \V(G')  c(v) = c_{1}(v),
 \forall i\in\{1,\ldots,2t+2\} c(v_{i}) = c_{1}(v),
 \forall i\in\{2t+3,\ldots,4t+4\}c(v_{i}) = c_{2}(v_{i}),
 \forall v \in \V(G'''')  c(v) = c_{2}(v)$.
Moreover, the vertices $v_{t+2},\ldots,v_{3t+3}$ inherit the common certificates. Then all nodes of $G'$ with certificate $c$ have isomorphic $t$-neighborhoods as in $G_{t}(G',v',0,G'',v'',1)$ with certificate $c_{1}$. Moreover, all nodes of $G''''$ with certificate $c$ have isomorphic $t$-neighborhoods  as in $G_{t}(G''',v''',1,G'''',v'''',0)$ with certificate $c_{2}$. 
The nodes $v_{1},\ldots,v_{2t+2}$ of $P(t)$ with certificate $c$ have isomorphic $t$-neighborhoods as in $G_{t}(G',v',0,G'',v'',1)$ with certificate $c_{1}$.  The nodes $v_{2t+3},\ldots, v_{4t + 4}$ with certificate $c$ have isomorphic $t$-neighborhoods as in $G_{t}(G''',v''',1,G'''',v'''',0)$ with certificate $c_{2}$.
Thus, all  nodes must accept the configuration $G_{c} \notin L$. This is a contradiction.
\end{proofof}

\subsection{A Linear Lower Bound for the Certificate Size of Fixed Point Free Symmetries}\label{sec:linear}

  In order to prove Lemma \ref{satz:insideNLD:FixedPointFreeSymmetryOnTrees}, we will focus on the family $ \setT_{n}$ of  all connected trees with $n$ nodes and without any isomorphic copies.  Therefore, we will need to know the number of trees with $n$ nodes. 
 For $n\geq 2$ holds: there are $n^{n-2}$ trees with $n$ nodes\cite{neville1953codifying}. 
Since a graph with $n$ nodes is at most isomorphic to $n!$ graphs, we have
$ |\setT_{n}| \geq \frac{n^{n-2}}{n!} \text{ and } \log\left(\frac{n^{n-2}}{n!}\right) = \Theta(n).$

The construction of test graphs is as follows: these graphs are built out of two trees from $\setT_{n}$ in the following way: let  $T',T'' \in \setT_{n}$ be two trees and let $v'\in \V(G')$ and $v''\in \V(G'')$ be two nodes of the graphs. Then $T_{n}(T',v',T'',v'')$  consists of the two trees $T', T''$ and an additional path $v',v_{1},\ldots, v_{2t+2},\ldots,v_{\psi(n)},v''$ which we will denote by $P(n)$, whereby  $$ 
  \psi(n) = \begin{cases}
     n    \qquad&,\text{ if }n\text{ is even}\\
     n-1 &,\text{ if } n\text{ is odd}
  \end{cases}
$$
ensures an even size of the path. This is of importance when we will look at the fixed point free symmetries of the graph. 
 The nodes do not have inputs. An example can be found in Figure \ref{fig:lowerbound:tree}. 
We need the following property of the constructed tree $T_{n}(T',v',T'',v'')$ for the lower bound: 
for $T',T'' \in \setT_{n}$we have: $T_{n}(T',v',T'',v'') \in \FixedPointFreeSymmetryOnTrees$ if and only if $T' = T''$.
We can show Lemma \ref{thm:lowerbound:fixedpointfreesymmetryontrees} like before.

\section{Some Open Problems}
 One interesting question is whether there are languages which need certificates in between $\log n$ and $n$ or $n$ and $n^2$. Since  $\NLD = \NLDf{n^2 + |w|}$ holds, another question is whether there is a language which needs certificates of size $\Theta(n^2 + |w|)$.  
Furthermore, it is interesting to know, whether there are trade offs between running time and certificate sizes, if the requirement of a constant number of communication rounds is omitted.



\bibliographystyle{plain}
\bibliography{literatur}


\pagebreak \small

\begin{appendix}
\section{Additional Figures}

\begin{figure}[h!!!]
  \begin{center}
     \includegraphics[width=0.3\textwidth]{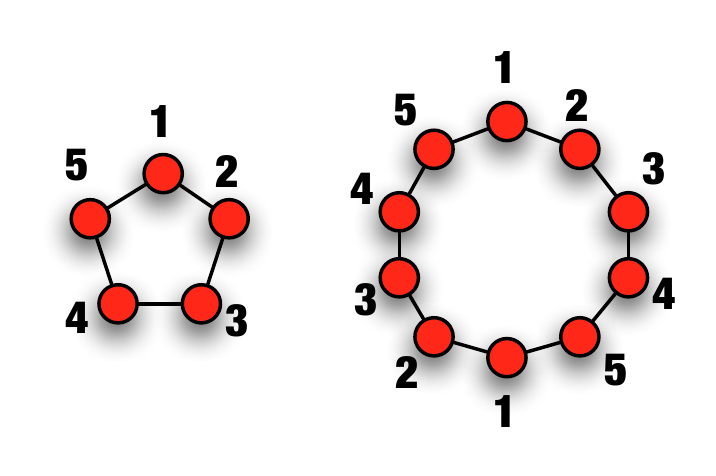}\\
     \includegraphics[width=0.6\textwidth]{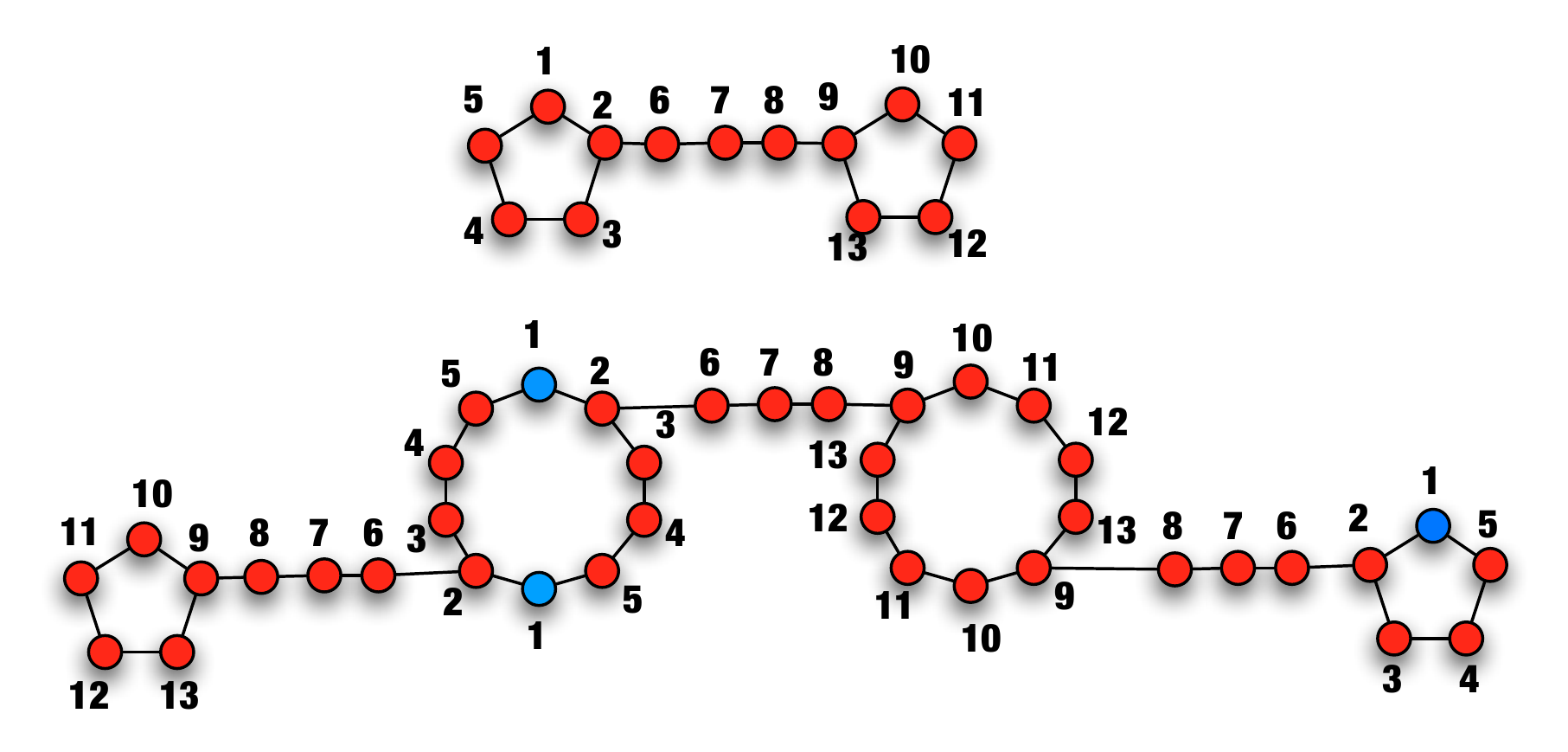}
  \end{center}
   \caption{The figure shows two pairs of a lift and a lifted configuration. The labels are the assigned identifiers. The first example of a lift is a cycle. The lift occurs two times in the lifted configuration.  The second example is more complex. The lift occurs three times in the lifted configuration. We will see that it is no coincidence that the graphs have cycles. }
   \label{fig:lifts:examples}
\end{figure}

\begin{figure}[h!!!!]
\begin{center}
  \includegraphics[width=0.4\textwidth,angle=90]{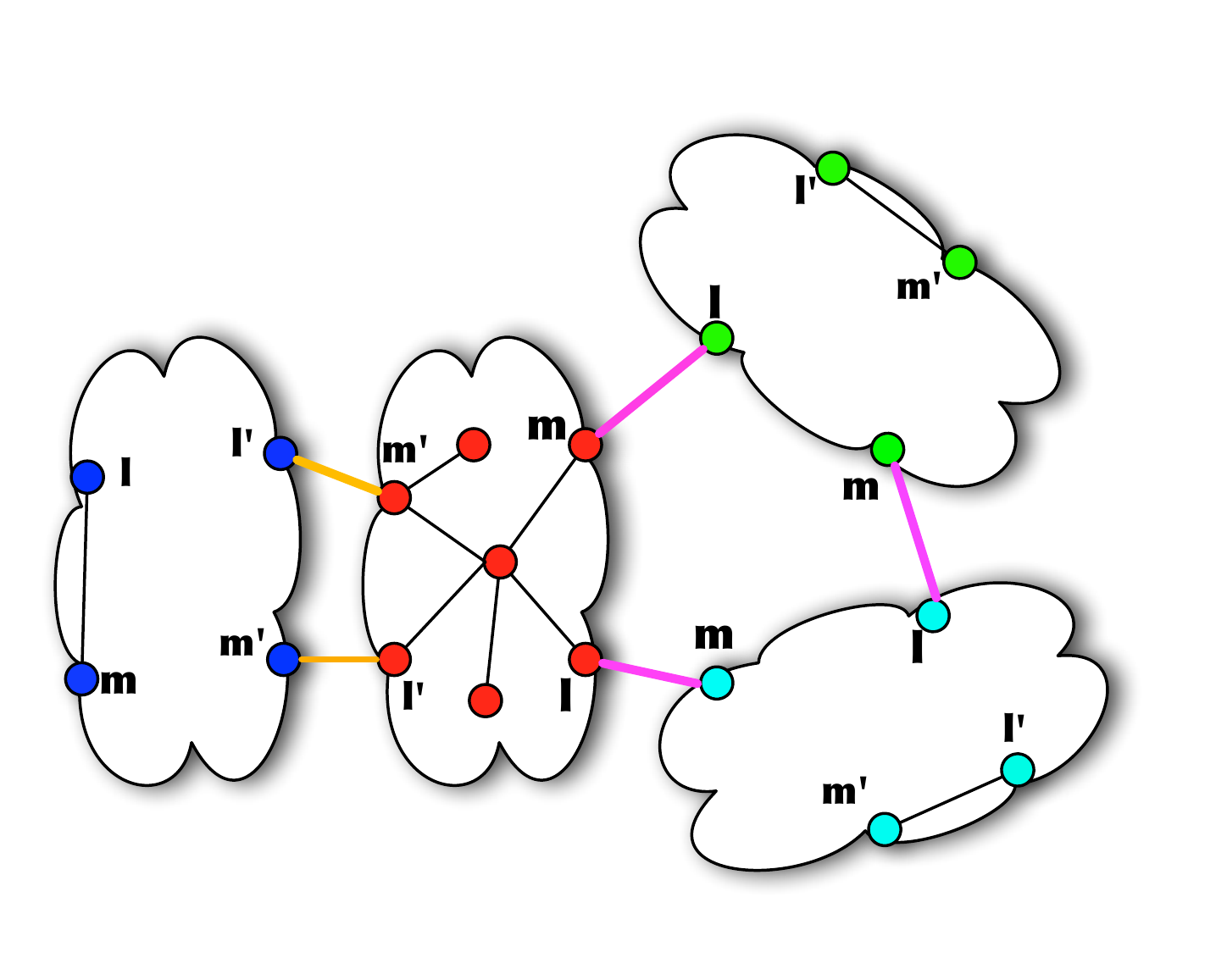}
  \end{center}
  \caption{The figure shows the two possibilities for the defined sequences. The orange edges show the case of a cycle between two partitions. The purple edges show the case of a cycle with several partitions.}
  \label{fig:proof:lift:cycle}
\end{figure}

\begin{figure}
  \begin{center}
      \includegraphics[width=0.45\textwidth]{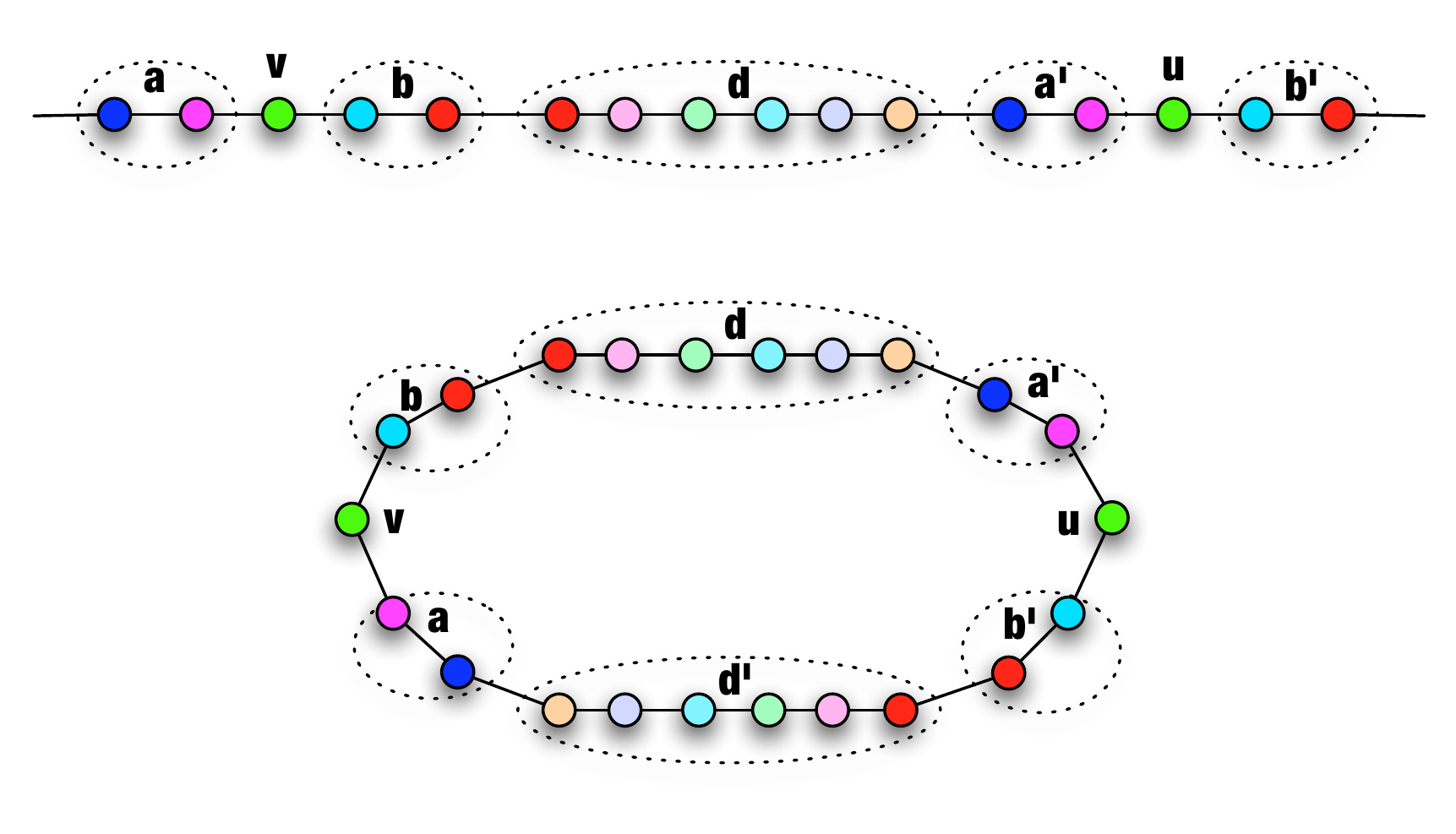}
  \end{center}
   \caption{The figure shows the construction of the cycle from the path. The colors of the nodes denote the different certificates.}
   \label{fig:lowerbound:tree:cycleconstruction}
\end{figure}

\begin{figure}[h!!!tbp]
   \begin{center}
      \includegraphics[width=0.5\textwidth]{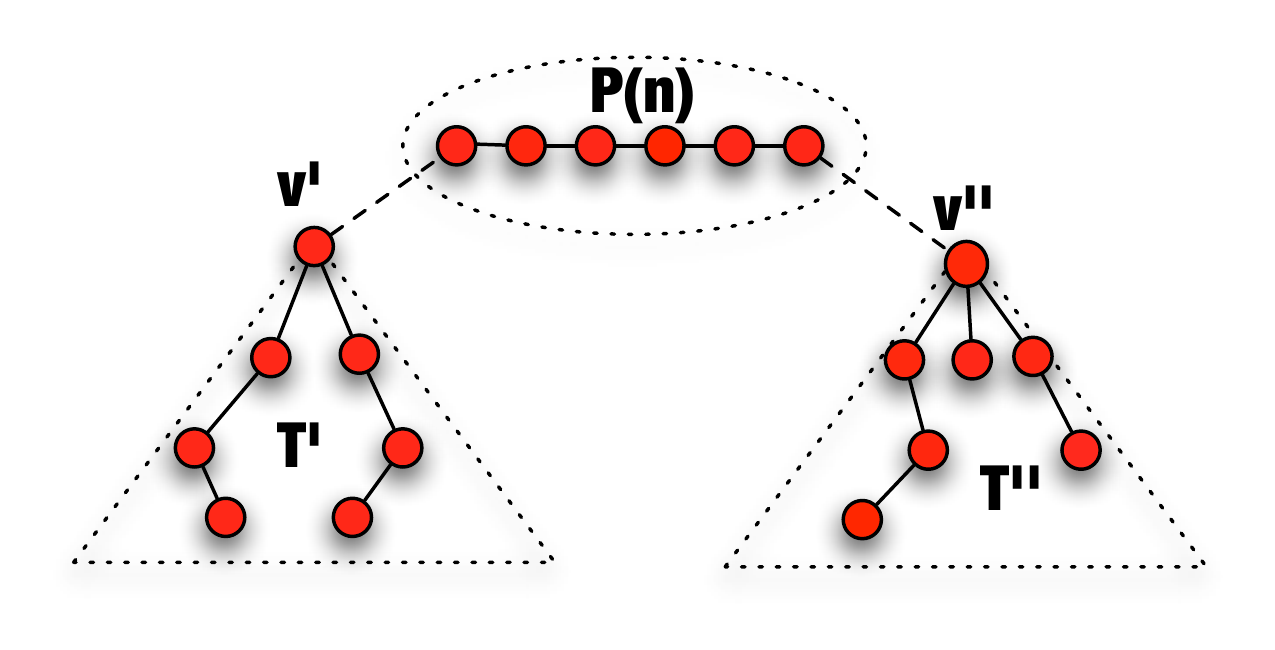}
   \end{center}
  \caption{The figure shows the tree $T_{7}(T',v',T'',v'')$ which is built out of two trees of size $7$. It is noteworthy that $P(n)$ only consists of $6$ nodes as we make sure that the number of nodes of the path is even. Since the two trees are not isomorphic, the tree does not belong to the language $\FixedPointFreeSymmetryOnTrees$.}
  \label{fig:lowerbound:tree}
\end{figure}

\begin{figure}[h!!!]
   \begin{center}
     \includegraphics[width=0.5\textwidth]{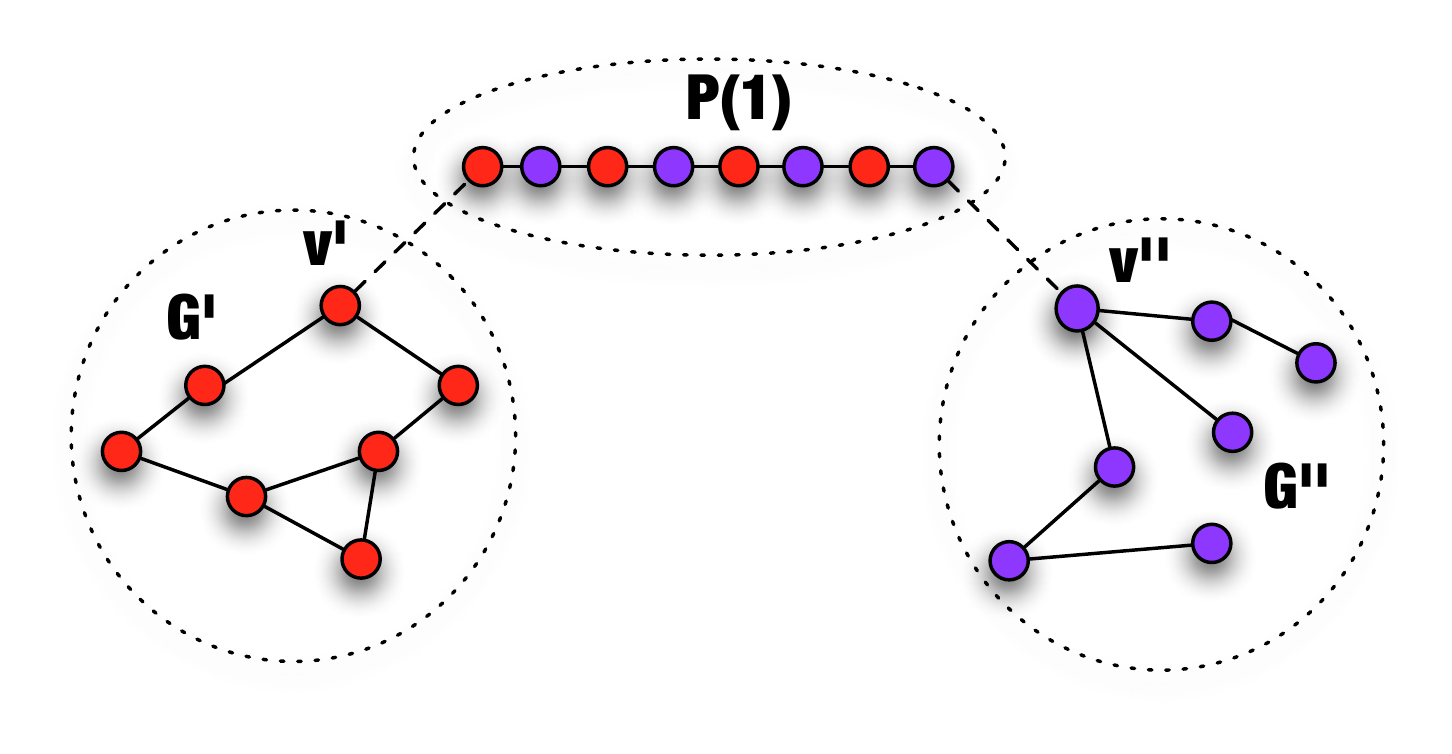}
   \end{center}
   \caption{The figure shows the graph $G_{t}(G',v',0,G'',v'',1)$ for two graphs $G',G''$. The two partitions are visualized by the two different colors. In this example, the partitions are equally sized and, thus, the configuration belongs to the language $\EqualSizePartition$. }
   \label{fig:lowerbound:PartionGraph}
\end{figure}

\newpage
\section{An Overview of Certificate Sizes}\label{sec:appendix:certsize}

The following table gives an overview of the sizes of  certificates that are needed to verify the languages. The proofs are given in \cite{kamilmsc}. Most languages are not defined in this paper, but the languages should be understandable by their names. The $\Omega(1)$ lower bounds are given by proofs which show that a language cannot be decided deterministically.
\vskip 1cm
\begin{center}
 \begin{tabular}{|p{5.7cm}|p{2.1cm}|p{2.1cm}|}
 \hline
         &     \multicolumn{2}{c|}{ \textbf{certificate size bounds}}         \\ 
  \textbf{Name} &  \textbf{min} & \textbf{max}  \\
  \hline \hline
   $\Avg$ & $\Omega(1)$   & $\bO(n^{2} +|w| )$       \\
   $\AvgDeg$ &    & $\bO(n^{2} +|w| )$      \\
   $\Bipartite$ &  $\Omega(1)$  & $\bO(1 )$     \\
  $\Clique$ & $\Omega(\log n) $    & $\bO(\log n)$      \\
  $\kColorable$ &  $\Omega(1)$   & $\bO(1 )$       \\
  $\kConsensus$ & $\Omega(\log n)$     & $\bO(\log n)$     \\
  $\DomaticNumber$ &  $\Omega(1)$  & $\bO(|w| )$     \\
  $\kDominatingSetNotFixed$ &   $\Omega(1)$ & $\bO(\log n )$   \\
  $\kEdgeColorable$ &  $\Omega(1)$  & $\bO(\log n )$      \\
  $ \EqualSizePartition$        & $\Omega(n^{2})  $  & $\bO(n^{2}) $       \\
   $\FixedPointFreeSymmetryOnTrees$ & $\Omega(n)$ & $\bO(n)$       \\
   $\HasPerfectMatching$ &  $\Omega(1)$  & $\bO(\log n )$      \\
   $\LogicalAnd$ & $\Omega(1)$   & $\bO(\log n)$    \\
   $\LogicalOr$ & $\Omega(1)$   & $\bO(\log n )$     \\
   $\Max$  & $\Omega(\log n)$     & $\bO(\log n)$     \\
   $\coMaximumMatching$ &  $\Omega(1)$  & $\bO(\log n )$      \\
   $\Min$ & $\Omega(\log n) $    & $\bO(\log n)$   \\
   $\Mode$ &   $\Omega(1)$ & $\bO(n^{2} +|w| )$       \\
   $\MostDeg$ &    & $\bO(n^{2} +|w| )$     \\
   $\NonEmptySetIntersection$ &  $\Omega(1)$  & $\bO(|w| )$     \\
   $\kSmallest$ & $\Omega(\log n)$     & $\bO(\log n)$   \\
  $\Tree$   & $\Omega(\log n)$     & $\bO(\log n)$        \\  
  \hline
 \end{tabular}
\end{center}

\end{appendix}

\end{document}